\documentclass[journal]{IEEEtran}
\ifCLASSINFOpdf
\else
\fi
\usepackage{amsmath, amssymb, epsfig, graphicx, bm}
\usepackage{cases}
\usepackage{float, subfigure}
\usepackage{color}
\usepackage{amsthm}
\usepackage{amsfonts}
\usepackage{psfrag}
\usepackage{euscript}
\usepackage{cite, url}
\usepackage{breqn}
\usepackage[ruled]{algorithm2e}
\usepackage[all,cmtip]{xy}
\usepackage{color}
\definecolor{darkblue}{rgb}{0,0,1}
\usepackage{enumitem}

\newcommand{\g}{\text{g}}
\newcommand{\dt}{\text{d}}
\theoremstyle{plain}

\newtheorem{theorem}{Theorem}

\newtheorem{counter}{Counter Example}

\newcommand{\descref}[1]{%
	\ref{#1}*%
}

\definecolor{amber}{rgb}{1.0, 0.49, 0.0}

\theoremstyle{definition}

\newcommand{\m}{{\text{m}}}
\newcommand{\M}{{\text{M}}}

\newcommand{\argmin}[1]{\underset{#1}{\text{argmin}}\;}

\def\zb{\mathbf{z}}

\def\z*{\zb_t^*}

\begin{document}
%
\title{Revisiting "Consensus-Based Energy-Management in Smart Grid with Transmission Losses and Directed Communication"}
%
%
%

\author{Jan Zimmermann, Tatiana Tatarenko, Volker Willert and J\"urgen Adamy
	\thanks{The authors are with the Control Theory and Robotics Lab of the Department of Electrical Engineering, Technical University of Darmstadt, 64283, Germany.
	}
	\thanks{This work was funded by the Deutsche Forschungsgemeinschaft (DFG, German Research Foundation) - SPP 1984.}
} 
\maketitle

\begin{abstract}
We discovered a deficiency in Algorithm 1 and Theorem 3 of \cite{Zhao2017}. The algorithm called CEMA aims to solve an energy management problem distributively. However, by means of a counter example, we show that Theorem 2 and 3 of \cite{Zhao2017} contradict each other in the case of a valid scenario, proving that the suggested algorithm does not always find the optimum. Furthermore, we provide theoretic results, showing that Theorem~3 of \cite{Zhao2017} does not hold generally. At last, we provide a rectification by adjusting the algorithm and the corresponding proof of Theorem 3. 
\end{abstract}


%
\IEEEpeerreviewmaketitle


\section{Introduction}
In \cite{Zhao2017}, the authors present an incremental cost algorithm that is applied to a distributed economic dispatch problem in which power consumers can vary their demand in some range. All instances involved in the optimization process, power generators and consumers, are represented by nodes in a network. All these nodes are connected by directed communication channels to exchange information. The problem can be formally defined as follows:
\begin{subequations}\label{eq:optprob}
	\begin{align}
		\min & \sum_{i \in \mathcal{V}_\g} C_i(P_i) -  \sum_{j \in \mathcal{V}_\dt} U_j(P_j), \\
		\text{s.t.} & \sum_{i \in \mathcal{V}_\g} (P_i - B_iP_i^2) =  \sum_{j \in \mathcal{V}_\dt} P_j, \label{eq:powerbalance} \\
		 & P_i^\m \leq P_i \leq P_i^\M, i \in \mathcal{V}_\g, \label{eq:boundsgen}\\
		 & P_j^\m \leq P_j \leq P_j^\M, j \in \mathcal{V}_\dt.  \label{eq:boundscon}
	\end{align}
\end{subequations}
Here, $\mathcal{V}_\g$ denotes the set of all generator nodes and $\mathcal{V}_\dt$ the set of all consumer nodes. The entire node set can be described by their union $\mathcal{V} = \mathcal{V}_\g \cup \mathcal{V}_\dt$. In above formulation of the problem, the convex function $C_i(P_i)$ represents the $i$-th generators cost to produce the power $P_i$. Analogously, $U_j(P_j)$ is some concave utility function of the power consumer $j$. Transmission losses at the generators are considered by the factor $B_iP_i^2$ in the power balance constraint \eqref{eq:powerbalance}, where $B_i > 0$ and chosen such that $B_iP_i^2 < P_i$ in the range $P_i^\m \leq P_i \leq P_i^\M$. At last, in \eqref{eq:boundsgen} and \eqref{eq:boundscon}, lower and upper bounds are imposed on power generation and consumption, with $0 < P_i^\m$ denoting the minimal and $P_i^\M$ the maximal production or consumption of node $i$. \\
\begin{algorithm}[h!]\label{alg:rectalg1}
	\SetAlgorithmName{Algorithm 1 (1*)}{algorithm 1 (1*)}{List of Algorithms}
	\renewcommand{\thealgocf}{}
	\SetAlgoLined
	\KwResult{$\lambda_i(k)$, $P_i(k), \forall i \in \mathcal{V}$.}
	\textbf{Initialization:} Set $\lambda_i(0)$, $P_i(0)$ and $\xi_i(0)$, $\forall i \in \mathcal{V}$ as 
	\begin{align}
	\lambda_i(0) &= \begin{cases}
	\frac{C_i'(P_i^\m)}{(1 - 2 B_i P_i^\m)}, &i \in \mathcal{V}_\g, \\
	U_i'(P_i^\M), & i \in \mathcal{V}_\dt.
	\end{cases}\\
	P_i(0) &= 0, \ i \in \mathcal{V}.\\
	\xi_i(0) &= 0, \ i  \in \mathcal{V}. 
	\end{align}
	Set  $\eta$, $0 < \eta < 1$, and termination errors $\epsilon_m$ and $\epsilon_l$. \\
	\textbf{Iteration}: \\
	1) Update $\lambda_i$ according to 
	\begin{equation}
	\lambda_i(k+1) = \sum_{j \in \mathcal{V}} w_{ij} \lambda_j(k) + \eta \xi_i(k), \ i \in \mathcal{V}. \label{eq:lambdaupdate} 
	\end{equation} \\
	2) For $i \in \mathcal{V}_\g$, update $P_i$ according to
	\begin{align}
	P_i(k+1) &= \argmin{P_i^\m \leq P_i(k) \leq P_i^\M} \big[C_i(P_i(k)) \nonumber \\ 
	& - \lambda_i(k+1)  P_i(k)     \big]; \label{eq:oldline} 
	\end{align}\\
	2*) For $i \in \mathcal{V}_\g$, update $P_i$ according to
	\begin{align}
	P_i(k+1) &= \argmin{P_i^\m \leq P_i(k) \leq P_i^\M} \big[C_i(P_i(k)) \nonumber \\ 
	&- \lambda_i(k+1)\left(  P_i(k) - B_iP_i(k)^2  \right)   \big]; \label{eq:correcedline}\tag{\ref{eq:oldline}*}
	\end{align}
	For $i \in \mathcal{V}_\dt$, update $P_i$ according to
	\begin{align}
	P_i(k+1) &= \argmin{P_i^\m \leq P_i(k) \leq P_i^\M} \big[ \lambda_i(k+1)P_i(k) \nonumber \\ &
	- U_i(P_i(k)) \big];
	\end{align}\\
	3)  For $i \in \mathcal{V}_\g$, update $\xi_i$ according to 
	\begin{align}
	\xi_i(k+1) &= \sum_{j \in \mathcal{V}} q_{ij} \xi_j(k) + \left( P_i(k) - B_iP_i(k)^2 \right) \nonumber \\
	&- \left(P_i(k+1) - B_iP_i(k+1)^2\right); 
	\end{align}
	For $i \in \mathcal{V}_\dt$, update $\xi_i$ according to
	\begin{align}
	\xi_i(k+1) = \sum_{j \in \mathcal{V}} q_{ij} \xi_j(k) +  P_i(k+1) - P_i(k).
	\end{align}\\
	4) If $|\xi_i(k)|\leq \epsilon_m, \forall i \in \mathcal{V}$ and $|\lambda_i(k) - \lambda_i(k-1)| \leq \epsilon_l$, $ \forall i \in \mathcal{V}$, break.\\
	
	\vspace{0.2cm}
	\caption{CEMA (*: Adjusted CEMA) }
\end{algorithm}
Theorem 1 of \cite{Zhao2017} provides the convex reformulation
\begin{equation}
	\sum_{i \in \mathcal{V}_\g} (P_i - B_i P_i^2) \geq \sum_{j \in \mathcal{V}_\dt} P_j \label{eq:convexpowerbalance}
\end{equation}
of the non-convex equality condition in \eqref{eq:powerbalance} and proves that the convex problem and  the non-convex problem share the same optimal value  $P^*$. To solve the obtained convex problem, a consensus-based energy management algorithm (CEMA) is presented in \cite{Zhao2017} and can also be found as Algorithm 1 in this publication. In order to distinguish between the versions from \cite{Zhao2017} and our rectification, we use *. Therefore, equation \eqref{eq:oldline} refers to the update equation of \cite{Zhao2017} and \eqref{eq:correcedline} refers to our suggestion. 
 Theorem 2 of \cite{Zhao2017} provides the proof that Algorithm 1 converges to some vector $\hat{P}^*$. Then, in Theorem 3 of \cite{Zhao2017}, the goal is to  show that this final vector  $\hat{P}^*$ produced by Algorithm 1 is equal to the optimal vector $P^*$ of Problem \eqref{eq:optprob}. We refer the reader to the original paper \cite{Zhao2017} for further details. \\
In what follows, we will first show in Section II that Theorem 2 and Theorem 3 contradict each other in case of a valid counter example, i.e. that the results of Theorem 2 and 3 do not always hold. In Section III, we will show that the proof of Theorem 3 is incorrect in relation to the presented algorithm in \cite{Zhao2017}. In Section IV, we provide an adjustment of equation \eqref{eq:oldline}, resulting in equation \eqref{eq:correcedline}, and a rectification of Theorem 3.

\section{Counter example}
In this section, we provide a counter example that shows that under Algorithm 1, including equation \eqref{eq:oldline}, the results of Theorem 2 and Theorem 3 contradict each other. We choose a simple example of two generators and two consumers. The cost functions for the two node types are defined as in \cite{Zhao2017} with
\begin{align}
	C_i(P_i) &= a_i P_i^2 + b_i P_i + c_i, \\
	U_j(P_j) &= \begin{cases}
	w_j P_j - \alpha_j P_j^2, & P_j \leq \frac{w_j}{2\alpha_j}, \\
	\frac{w_j^2}{4\alpha_j}, & P_j > \frac{w_j}{2\alpha_j}.
	\end{cases}
\end{align}
All parameters defining the cost and utility functions are considered to be positive. Their parameterisation of the considered scenario is summarized in Table \ref{tab:paramgen}. 

\begin{table}[h]
	\centering
	\begin{tabular}{c|c|c|c|c|c|c}
		                                                 \multicolumn{7}{c}{Generator parameters}                                                   \\
		\rule{0pt}{10pt}	ID,$ i$ &  $a$   &         $b$          & $c$ &           $B$           & $P^{\text{m}}$ &                  $P^{\text{M}}$                  \\ \hline\hline
		           1             & 0,0024 &         5,56         & 30  &         0,00021         &  60   &                 339,69                  \\ \hline
		           2             & 0,0056 &         4,32         & 25  &         0,00031         &  25   &                 479,10                  \\
		                                                           \multicolumn{7}{c}{}                                                            \\
		                                                  \multicolumn{7}{c}{Consumer parameters}                                                   \\
		\rule{0pt}{10pt}	ID, $j$ & \multicolumn{2}{c|}{$\omega$} & \multicolumn{2}{c|}{$\alpha$} & $P^\m$ &                  $P^\M$                  \\ \hline\hline
		           1             &  \multicolumn{2}{c|}{18,43}   &  \multicolumn{2}{c|}{0,0545}  &  50   &                 100,34                  \\ \hline
		           2             &  \multicolumn{2}{c|}{13,17 }  &  \multicolumn{2}{c|}{0,0877}  &  100  & 159,13     
		           \vspace{0.1cm}
	\end{tabular} 
	\caption{Parameterization of generator and consumer agents.}
	\label{tab:paramgen}
\end{table}

Theorem 2 of \cite{Zhao2017} states that, using Algorithm 1, the power mismatch converges to zero:
\begin{equation}
	\lim_{k \rightarrow \infty} \left[\sum_{j \in \mathcal{V}_\dt} P_j(k) - \sum_{i \in \mathcal{V}_\g} (P_i(k)- B_iP_i^2(k))\right] = 0
\end{equation}
and the distributed Lagrange multipliers of the power mismatch $\lambda_i(k)$, namely the incremental cost or utility of agent $i$ updated by equation \eqref{eq:lambdaupdate}, converges to a consensus value:
\begin{equation}
	\lim_{k \rightarrow \infty}\lambda_i(k)  = \lambda_c, \ \forall i \in \mathcal{V}. 
\end{equation}
Theorem 3 provides proof that the result $\hat P^*$ of Algorithm 1  is equal to the optimal solution $P^*$ of the convex optimization problem. In what follows, we will briefly show that the convergence of $\lambda_i(k)$ to a consensus value $\lambda_c$ and the statement of Theorem 3 contradict each other for the scenario shown in Table \ref{tab:paramgen}. In our analysis, we will focus exclusively on the optimal values of the generators.
\begin{counter}
	Using Algorithm 1 and the valid scenario in Table \ref{tab:paramgen}, the statements of Theorem 2 and 3 in \cite{Zhao2017} contradict each other. 
\end{counter}
\begin{proof}
Solving the problem defined in \eqref{eq:optprob} using parameters of Table \ref{tab:paramgen} and an appropriate solution method, we receive the optimal values 
\begin{align}
	P_\g^* = \begin{bmatrix}
	81.98, 124.80
	\end{bmatrix}^T.
\end{align}
Assuming that all $\lambda_i(k)$  converge to $\lambda_c$, i.e. Theorem 2 is true,  and the power mismatch is $0$, the optimal value of the algorithm can be calculated as
\begin{equation}
	\hat{P}_i^* = \argmin{P_i}  C_i(P_i) - \lambda_c P_i, \ i \in \mathcal{V}_\g,
\end{equation}
according to equation \eqref{eq:oldline} of Algorithm 1. A necessary and sufficient condition for some value $P_i$ to be the optimum of the problem is that $P_i$ renders the first derivation of above objective function to zero. 
Therefore, setting the first derivative to zero and reordering, we receive
\begin{equation}\label{eqn:lambda_c}
	\lambda_c = \frac{\delta C_i(\hat{P}_i)}{\delta \hat P_i}\bigg|_{\hat P_i = \hat P_i^*} = 2 a_i \hat{P}_i^* + b_i.
\end{equation}
Now, if we assume that the solution of the algorithm and the solution of the convex problem are identical $\hat{P}^* = P^*$, i.e. the statement of Theorem 3 is true, we receive the following values for $\lambda_c$ by inserting the corresponding parameters and optimal values in \eqref{eqn:lambda_c}:
\begin{align}
	\text{Generator 1: }  &  \lambda_c^1 = 5.95,\\
	\text{Generator 2: }  &  \lambda_c^2 = 5.72
\end{align}
Therefore,  $\lambda_c^1 \neq \lambda_c^2$, which contradicts the statement of Theorem 2. 
\end{proof}

\section{Analysis of the Proof of Theorem 3}
By further inspection, it can be shown that the proof of Theorem 3 is erroneous, which means that, using Algorithm 1, the received solution $\hat P^*$ does not match the solution $P^*$ of Problem \eqref{eq:optprob}. To be more precise, the proof is sound for the update equations of the consumers, but does not hold for the generators.  Therefore, we focus in the following analysis on the proof for the generators  and show that
\begin{theorem}\label{thm:counter}
	Even if all assumptions of Theorem 3 of \cite{Zhao2017} hold, the solution $\hat{P}_i^*$, $i \in \mathcal{V}_\g$, obtained by Algorithm 1, does not satisfy the necessary and sufficient Karush-Kuhn-Tucker (KKT) conditions of the convex problem, given  $B_i > 0$ and $\hat{P}_i^*> 0$. Therefore, Theorem 3 of \cite{Zhao2017} does not hold generally. 
\end{theorem}
\begin{proof}
	The convex problem is defined in \eqref{eq:optprob}, where constraint \eqref{eq:powerbalance} is replaced with the convex constraint \eqref{eq:convexpowerbalance}.
	For the sake of understanding, we present the Lagrangian of the convex optimization problem and the necessary and sufficient KKT conditions from \cite{Zhao2017}:
	\begin{align*}
		\mathcal{L}(P, \lambda, \gamma, \nu) &= \sum_{i \in \mathcal{V}_\g} C_i(P_i) - \sum_{j \in \mathcal{V}_\dt} U_j(P_j) \\
		& + \lambda\left(\sum_{j \in \mathcal{V}_\dt}P_j + \sum_{i \in \mathcal{V}_\g}(B_iP_i^2 - P_i)\right) \\
		& + \sum_{i \in \mathcal{V}} \gamma_i(P_i^\m - P_i) + \sum_{i \in \mathcal{V}} \nu_i (P_i - P_i^\M).
	\end{align*}
	The corresponding KKT conditions read as follows:
	\begin{subequations}\label{eq:KKT-conditions}
	\begin{align}
		  &\frac{\delta C_i(P_i)}{\delta P_i}\bigg|_{P_i = P_i^*} - \lambda (1- 2B_iP_i^*) - \gamma_i + \nu_i = 0,                                     \forall i \in \mathcal{V}_\g \label{eq:lagranezerogen}\\
		  &\lambda - \gamma_j + \nu_j - \frac{\delta U_j(P_j)}{\delta P_j}\bigg|_{P_j = P_j^*} = 0,                                                     \forall j \in \mathcal{V}_\dt, \\
		  &\lambda \Bigg(\sum_{j \in \mathcal{V}_\dt}  P_j^* + \sum_{i \in \mathcal{V}_\g}\left(B_i(P_i^*)^2 - P_i^* \right) \Bigg) = 0,  \lambda \geq 0,                                               \\
		  &\gamma_i(P_i^\m - P_i^*) = 0,   \gamma_i \geq 0,                                                                                     \forall i \in \mathcal{V},                           \\
		  &\nu_i (P_i^* - P_i^\M) = 0,                                                                                               \nu_i \geq 0,             \forall i \in \mathcal{V}.
	\end{align}
\end{subequations}
	We will show now  that the result $\hat P_i^*, \ i \in \mathcal{V}_\g$ of the Algorithm 1  does not always render the first derivative of the Lagrangian function of the convex problem to zero, i.e. condition \eqref{eq:lagranezerogen} does not hold.\\ 
	Notice that with the result $\lim_{k \rightarrow \infty} \lambda_i(k) = \lambda_c$, $i \in \mathcal{V}_\g$ of Theorem 2 of \cite{Zhao2017}, we receive 
	\begin{align*}
		& \lim_{k \rightarrow \infty}P_i(k) =\\
		& \argmin{P_i^\m \leq P_i(k) \leq P_i^\M} \left[ C_i(\lim_{k \rightarrow \infty} P_i(k))  - \lambda_c \lim_{k \rightarrow \infty} P_i(k) \right] = \hat{P}_i^*.
	\end{align*}
	as the limit of equation \eqref{eq:oldline}.
	The resulting final value of above expression is equal to the solution of the optimization problem
	\begin{subequations}
		\begin{align}
			&\hat{P}_i^* = \argmin{P_i} C_i(P_i) - \lambda_c P_i \\
			&\text{s.t. } P_i^\m \leq P_i \leq P_i^\M. 
		\end{align}
	\end{subequations}
	Now assume that $P_i^\m < \hat{P}^* < P_i^\M$ holds for the optimal solution. After setting the first derivative to zero, we receive
	$\lambda_c = \frac{\delta C_i(\hat{P}_i)}{\delta \hat{P}_i}\big|_{\hat P_i = \hat P_i^*}$. Inserting this result in the KKT condition \eqref{eq:lagranezerogen}, it follows that
	\begin{equation}
		\frac{\delta C_i(\hat{P}_i)}{\delta \hat{P}_i}\bigg|_{\hat P_i = \hat P_i^*} 2B_i \hat{P}_i^* \overset{!}{=} 0.
	\end{equation}
	With this, the KKT condition \eqref{eq:lagranezerogen} only holds for some special cases, where either the optimum of all $C_i(P_i)$, $i \in \mathcal{V}$ coincide with the optimal solution of Problem \eqref{eq:optprob} or where $\hat{P}_i^* = 0$. 
	Therefore, for the general case the necessary condition \eqref{eq:lagranezerogen} does not hold.
\end{proof}

\section{Addition to Algorithm 1}

From the insight of the counter example and the findings of the previous section, we suggest the following addition to \cite{Zhao2017}'s algorithm, where we add the term $-B_iP_i(k)^2$ to the update equation for $P_i(k+1), \ \forall i \in \mathcal{V}_\g$ in equation \eqref{eq:oldline}  and receive equation \eqref{eq:correcedline}, resulting in Algorithm \descref{alg:rectalg1}. 
Furthermore, we extend Theorem 3 of \cite{Zhao2017} to the case of described addition.

\begin{theorem}
	Provided that
		$\sum_{j \in \mathcal{V}_\dt} P_j^\M \geq \sum_{i \in \mathcal{V}_\g} \big( P_i^\m - B_i(P_i^\M)^2 \big),$
	 i.e., the solution of the reformulated convex problem equals the solution of the non-convex problem in \eqref{eq:optprob}, the rectified Algorithm \descref{alg:rectalg1}, including update equation \eqref{eq:correcedline}, guarantees that
	$	\lim_{k \rightarrow \infty} P(k) = \hat{P}^* = P^*$.

\end{theorem}
\begin{proof}
	As the changes from Algorithm 1 to Algorithm \descref{alg:rectalg1} only affect the generators, we adopt the proof for the consumers from \cite{Zhao2017}. 
	For the generators,  we consider again the converged case as we did for the proof of Theorem \ref{thm:counter}, i.e. we focus on the optimization problem 
	\begin{subequations}
		\begin{align}
			&\min_{P_i} C_i(P_i) - \lambda_c (P_i - B_iP_i^2), \\
			&\text{s.t. } P_i^\m \leq P_i \leq P_i^\M,
		\end{align}
	\end{subequations}
	which is solved by every generator agent $i \in \mathcal{V}_\g$, see equation \eqref{eq:correcedline},
	and the corresponding Lagrangian function 
	\begin{align}
		\mathcal{L}_i(P_i, \gamma_i, \nu_i) &= C_i(P_i) - \lambda_c (P_i - B_iP_i^2)  \nonumber\\
		&+ \gamma_i(P_i^\m - P_i)  +\nu_i(P_i - P_i^\M). \label{eq:localLagrangian}
	\end{align}
	First, we deal with the case that $P_i^\m < \hat{P}_i^* < P_i^\M$. Then $\gamma_i = \nu_i = 0$ and 
	\begin{equation}
		\lambda_c = \frac{\frac{\delta C_i(P_i)}{\delta P_i}\big|_{P_i = P_i^* }}{(1-2B_iP_i^*)}.
	\end{equation}
	
	Setting $\lambda = \lambda_c$ and inserting above result for $\lambda_c$ in KKT condition \eqref{eq:lagranezerogen}, together with $\gamma_i = \nu_i = 0$, shows that the KKT conditions for the convex problem are satisfied for this case.\\
	Consider now the case that $\hat{P}_i^* = P_i^\m$ ($\hat{P}_i^* = P_i^\M$). Then, we have $\gamma_i \geq 0, \ \nu_i = 0$ ($\gamma_i = 0, \nu_i \geq 0$). Setting the first derivative of the Lagrangian \eqref{eq:localLagrangian} to zero and solving for $\lambda_c$ results in $\lambda_c = \left(\frac{\delta C_i(P_i)}{\delta P_i}\big|_{P_i = P_i^*} - \gamma_i \right) / (1- 2B_i \hat P_i^*)$  $\left( \lambda_c = \left(\frac{\delta C_i(P_i)}{\delta P_i}\big|_{P_i = P_i^*} + \nu_i \right) / (1- 2B_i \hat P_i^*)\right)$. Setting $\lambda = \lambda_c$ and inserting $\lambda_c$ into the KKT condition \eqref{eq:lagranezerogen}, one can verify that the KKT conditions for the convex problem are met by $\lambda_c$ and $ \hat P_i^*$,  $\forall i \in \mathcal{V}_\g$. Therefore, together with the proof for the consumers in \cite{Zhao2017}, it holds that
	\begin{equation}
		\lim_{k \rightarrow \infty} P_i(k)	= \hat P_i^* = P_i^*, \ \forall i \in \mathcal{V}.  
	\end{equation}
\end{proof}

\ifCLASSOPTIONcaptionsoff
  \newpage
\fi

\end{document}